\documentclass[conference,a4paper]{IEEEtran}
\usepackage{cite}
\usepackage{amsmath,amssymb,amsfonts,amsthm}
\usepackage{graphicx}
\usepackage{textcomp}
\usepackage{booktabs}
\usepackage{url}
\usepackage{eso-pic}
\AddToShipoutPictureFG*{\AtPageLowerLeft{\put(\LenToUnit{1.6cm},\LenToUnit{0.9cm}){\parbox[b]{\dimexpr\paperwidth-3.2cm\relax}{\centering\footnotesize\itshape This work has been submitted to the IEEE for possible publication. Copyright may be transferred without notice, after which this version may no longer be accessible.}}}}

\newcommand{\herm}{^{\mathsf{H}}}
\newcommand{\vect}[1]{\mathbf{#1}}
\theoremstyle{plain}
\newtheorem{proposition}{Proposition}

\begin{document}

\title{Placing the Beamspace Window on a Hopping Carrier: Bin Drift, Exact Capture and the Cost of Rounding}

\author{
\IEEEauthorblockN{Moris S. Mukwayi\IEEEauthorrefmark{1}, Michael A. van Wyk\IEEEauthorrefmark{1}, Andre M. McDonald\IEEEauthorrefmark{2}}
\IEEEauthorblockA{\IEEEauthorrefmark{1}School of Electrical and Information Engineering, University of the Witwatersrand, Johannesburg 2000, South Africa}
\IEEEauthorblockA{\IEEEauthorrefmark{2}Defence and Security Cluster, Council for Scientific and Industrial Research, Pretoria 0184, South Africa}
\IEEEauthorblockA{Contact: moris.qed@gmail.com (corresponding), anton.vanwyk@wits.ac.za, amcdonald@csir.co.za}
}

\maketitle

\begin{abstract}
\hyphenpenalty=10000 \exhyphenpenalty=10000
A frequency-agile radar with a large digital array moves the spatial frequency of its look direction across the beams of a spatial FFT every time it changes carrier, so a beamspace adaptive processor must decide where to put its window at each hop. Because the beamspace transform is orthonormal, the output of the processor depends only on the span of the window, and a hop changes nothing but where the window sits. Centring the window exactly on the fractional bin of the look direction keeps the look vector fully inside it at every carrier, for the price of one phase ramp before the FFT. On a 32 by 32 array over 4 to 8~GHz the look bin drifts by 1.6 to 6.9 bins for look directions from u=0.2 to 60 degrees. Rounding to the nearest bin costs up to 1.34, 0.74 and 0.50~dB of look-direction capture for 3 by 3, 5 by 5 and 7 by 7 windows. In a three-jammer scenario with 41 carriers, a window frozen at the lowest carrier loses 39 to 60~dB at the worst carrier for a look direction of u=0.55, the exact window beats nearest-bin placement by 0.31 to 0.65~dB there, and the advantage reverses for a nine-beam window near broadside, where nearest-bin was 0.61~dB better.
\end{abstract}


\section{Introduction}

Frequency agility is a standard electronic protection measure: changing the carrier from dwell
to dwell denies a jammer a fixed band to concentrate on \cite{zhu2024}. On a large digital array the price is
paid in adaptive processing. Adapting over all $N$ elements by sample matrix inversion takes roughly $2N$ independent snapshots to land within 3~dB of the optimum \cite{rmb1974}, at a cost that grows as $N^3$, which is why reduced-dimension processing, in which a fixed transform compresses the array into a handful of channels and adaptation happens in that small space
\cite{vanveen1988,vantrees2002,somasundaram2015}, matters for large arrays and is still being
developed for the limited training support of airborne radar \cite{song2022,cao2024}. For massive digital radar arrays, adapting only a window of spatial-FFT beams around each target has been shown to approach full-dimensional MVDR at far lower complexity \cite{noroozi2025}, and the approach has been extended to a tiled 1024-element array \cite{noroozi2025b}. That work channelises the wideband signal into subbands, centres a small window on the FFT bin of each target, and studies the trade between window size and FFT zero padding. Windowed beamspace
processing has also been developed for massive multiuser MIMO reception
\cite{han2024,cebeci2024}. A hopping radar changes its carrier by a large fraction of the band from one dwell to the next, so the look direction moves across several FFT bins, and where to put the window becomes the central question.

The carrier dependence of the look bin is the beam squint, or spatial-wideband, effect of phased
arrays, reviewed for broadband mobile communication in \cite{abdalrazak2023}. Squint also moves
the nulls that an array places on interferers \cite{javed2023}, which bears on interference
suppression after a hop. Rounding the window to the FFT grid is the spatial counterpart of the
scalloping loss of a DFT bin \cite{harris1978}. For a uniform linear array and a window on the
nearest bin, \cite{cebeci2025} bound the energy that the window captures from a source off the
grid, independently of the array size, and show that zero padding can lower the SINR when
interference dominates; their wideband MIMO-OFDM study models the carrier dependence of the spatial
frequency and uses one window per subcarrier. Adaptive beamforming against jammers continues to be developed
for wideband, distributed and frequency-diverse arrays \cite{ma2024,han2023,nysaeter2022}. As far as we can tell, none of this work states where to place a fixed-size beamspace window when
the carrier of a radar hops, which is the question here.

The same difficulty is known from wideband direction finding, where covariance matrices at
different frequencies must be aligned before they are combined, and where the focusing
transforms that do it depend on preliminary direction estimates
\cite{wangkaveh1985,hungkaveh1988}. Here no estimate is needed, because the look direction is
known exactly and the geometry fixes where its bin lies at each carrier. The question is what
to do with that knowledge. This paper answers it for a beamspace MVDR window on a uniform
rectangular array. The invariance of Proposition~1 and the shift-theorem implementation are standard; what the paper adds is the placement question under hopping, the exact capture loss of a rounded window on a rectangular array, and the measurement of the placement rules against jammers. Section~II shows that the output of the processor depends only on the span
of the window and quantifies the drift of the look bin. Section~III evaluates the cost of rounding the window to the FFT grid, and of zero padding, for a rectangular array. Section~IV measures the placement rules against three jammers and over random jammer laydowns, and Section~V discusses what the results do and do not show.

\section{Window Position Under Hopping}

\subsection{Array and bin drift}
We consider a $32\times32$ uniform rectangular array ($N_x=N_y=32$, $N=1024$) whose elements are $d=18.75$~mm apart, which is half a wavelength at 8~GHz. A plane wave arriving from direction cosines $(u,v)$ at carrier $f$ gives element $(m,n)$ the phase factor
\begin{equation}
a_{mn}(u,v;f)=\exp\!\left(j\,\frac{2\pi f d}{c}\,(mu+nv)\right),
\label{eq:steer}
\end{equation}
$m,n=0,\dots,31$. A spatial FFT places this vector at the fractional bin
\begin{equation}
p^\star(f)=\frac{u_sN_xdf}{c},\qquad q^\star(f)=\frac{v_sN_ydf}{c},
\label{eq:bin}
\end{equation}
which grows linearly with the carrier. Over a band from $f_{\min}$ to $f_{\max}$ the look bin
drifts by $\Delta p=u_sN_xd\,(f_{\max}-f_{\min})/c$, which is $8.0\,u_s$ bins for this array
over 4 to 8~GHz (Table~\ref{tab:drift}). At the edge of a $\pm60^\circ$ sector the look bin
moves by almost seven bins, more than the half-width of any window that is small enough to be
worth using.

\begin{table}[!t]
\caption{Look-bin drift over 4 to 8~\textnormal{GHz}, $32\times32$ array (bins).}
\label{tab:drift}
\centering
\footnotesize
\begin{tabular}{@{}lcccc@{}}
\toprule
$u_s$ & 0.20 & 0.55 & 0.85 & 0.866 ($60^\circ$) \\
\midrule
$p^\star$ at 4~GHz & 1.60 & 4.40 & 6.80 & 6.93 \\
$p^\star$ at 8~GHz & 3.20 & 8.81 & 13.61 & 13.87 \\
Drift $\Delta p$ & 1.60 & 4.40 & 6.80 & 6.93 \\
\bottomrule
\end{tabular}
\end{table}

\subsection{An orthonormal window depends only on where it sits}
A $W\times W$ window with real centre $(p_c,q_c)$ has the columns
\begin{equation}
[\vect{T}]_{(m,n),(k,l)}=\frac{1}{\sqrt{N}}\,e^{\,j2\pi\left[m(p_c+k)/N_x+n(q_c+l)/N_y\right]},
\label{eq:T}
\end{equation}
$k,l=-(W-1)/2,\dots,(W-1)/2$, $B=W^2$. The columns are orthonormal for every real centre,
because two columns differ by an integer number of bins in each axis. Let
$\vect{R}$ be the interference-plus-noise covariance, $\vect{a}$ the look vector, and
$\vect{R}_b=\vect{T}\herm\vect{R}\vect{T}$, $\vect{a}_b=\vect{T}\herm\vect{a}$. The beamspace
MVDR processor adapts $\vect{w}_b=\vect{R}_b^{-1}\vect{a}_b/(\vect{a}_b\herm\vect{R}_b^{-1}\vect{a}_b)$.

\begin{proposition}[Invariance]
Let $\vect{T}'=\vect{T}\vect{G}$ with $\vect{G}$ unitary. Then the output SINR of the processor
with the true covariance, and the output SINR of sample matrix inversion with the same training
snapshots and diagonal loading $\sigma_{\mathrm{DL}}^2\vect{I}$, are the same for $\vect{T}$ and
$\vect{T}'$.
\end{proposition}

\begin{IEEEproof}
$\vect{R}_b'=\vect{G}\herm\vect{R}_b\vect{G}$, $\vect{a}_b'=\vect{G}\herm\vect{a}_b$, and every
training snapshot transforms by $\vect{G}\herm$. A diagonal loading matrix is unchanged,
$\vect{G}\herm(\sigma_{\mathrm{DL}}^2\vect{I})\vect{G}=\sigma_{\mathrm{DL}}^2\vect{I}$. The
adapted weight becomes $\vect{G}\herm\vect{w}_b$, and the effective element-space weight
$\vect{T}'\vect{G}\herm\vect{w}_b=\vect{T}\vect{w}_b$ is unchanged.
\end{IEEEproof}

A hop therefore changes the processor only through the span of the window, that is, through
its centre $(p_c,q_c)$: the carrier moves the look vector to a new bin by \eqref{eq:bin}, and
the question is where to put $(p_c,q_c)$. Four rules are compared:
\emph{fractional}, $(p_c,q_c)=(p^\star,q^\star)$; \emph{nearest-bin}, the integers nearest to
$(p^\star,q^\star)$, which is what a fixed FFT grid provides; \emph{frozen}, the nearest bin
at the lowest carrier, held for every hop; and a \emph{zero-padded} grid with spacing $1/P$ bins
and the nearest padded bin.

With the fractional rule, $\vect{a}$ is exactly $\sqrt{N}$ times column $(0,0)$, so
$\vect{a}_b=\sqrt{N}\,\vect{e}_0$ at every carrier: the look vector is captured completely,
whatever the hop. The rule is implemented by the shift theorem: multiply the snapshot by the
phase ramp $\exp[-j2\pi(mp_c/N_x+nq_c/N_y)]$, which separates into $N_x+N_y$ exponentials per
hop, take the unchanged FFT, and keep the central $W\times W$ bins. The ramp adds $N=1024$
complex multiplications per snapshot to an FFT that needs about $(N/2)\log_2N=5120$.

\section{The Cost of Rounding}

With the window centred on an integer bin and the look bin offset $(\delta_x,\delta_y)=(p^\star-p_c,q^\star-q_c)$,
the share of the look vector's energy that the window captures is
\begin{equation}
\eta=\eta_x(\delta_x)\,\eta_x(\delta_y),\qquad
\eta_x(\delta)=\sum_{k=-(W-1)/2}^{(W-1)/2}|D_{N_x}(\delta-k)|^2,
\label{eq:eta}
\end{equation}
with $D_M(x)=\sin(\pi x)/[M\sin(\pi x/M)]$, and the capture loss is $-10\log_{10}\eta$. In one dimension $\eta_x$ is the energy capture of a
nearest-bin window studied in \cite{cebeci2025}, which bounds it below, independently of $N_x$, by
replacing $|D_{N_x}|$ with a sinc function; the product form and the exact value at $N_x=32$ are what
this section adds for the rectangular array. For the
fractional rule $\delta=0$, $D_M(-k)=0$ for integer $k\ne0$, and $\eta=1$ exactly. Evaluating
\eqref{eq:eta} on the $32\times32$ array, and checking it against the explicit projection of
the 1024-element look vector onto the window (they agree to six digits), gives the worst case,
at an offset of half a bin in both axes, and the mean over offsets uniform in
$[-\tfrac12,\tfrac12]^2$ (Table~\ref{tab:round}). Rounding costs up to 1.34~dB for a $3\times3$
window and 0.50~dB for a $7\times7$ window before any interference is present, and the loss
follows a sawtooth in the carrier whose period is the carrier change that moves the look bin by one
bin (Figure~\ref{fig:drift}). Equation~\eqref{eq:eta} holds for any offset, so it also gives the capture loss of an error in the look direction: half a bin is a direction-cosine error of 0.031 at 8~GHz and 0.0625 at 4~GHz.

\begin{table}[!t]
\caption{Loss of look-direction capture (\textnormal{dB}) from rounding the window to the FFT grid, and from a zero-padded grid with the same $W\times W$ window.}
\label{tab:round}
\centering
\footnotesize
\setlength{\tabcolsep}{3.5pt}
\begin{tabular}{@{}ccccc@{}}
\toprule
 & \multicolumn{2}{c}{Nearest bin ($P=1$)} & \multicolumn{2}{c}{Zero padded, worst case} \\
\cmidrule(lr){2-3}\cmidrule(l){4-5}
$W$ & worst & mean & $P=2$ & $P=4$ \\
\midrule
3 & 1.345 & 0.694 & 0.014 & $<0.001$ \\
5 & 0.739 & 0.393 & $<0.001$ & $<0.001$ \\
7 & 0.504 & 0.271 & $<0.001$ & $<0.001$ \\
\bottomrule
\end{tabular}
\end{table}

\begin{figure}[!t]
\centering
\includegraphics[width=\columnwidth]{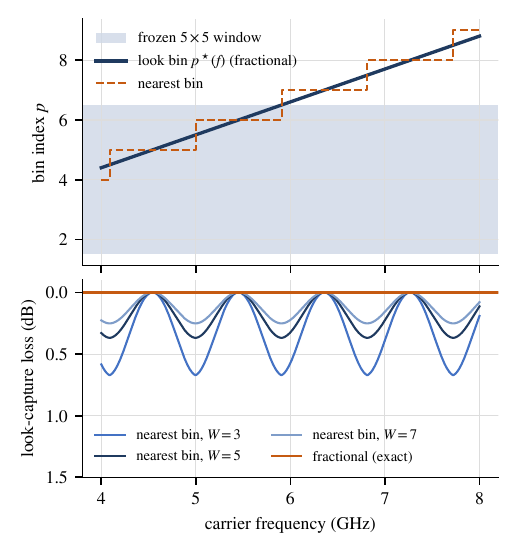}
\caption{Top: look bin of $u_s=0.55$ against carrier, with the nearest bin and the span of a $5\times5$ window frozen at 4~GHz; the look direction leaves that window at 5.91~GHz. Bottom: look-capture loss of the nearest-bin window; the fractional window has none.}
\label{fig:drift}
\end{figure}

Zero padding by $P$, which \cite{noroozi2025} studies as a design parameter and \cite{cebeci2025} find can lower the SINR when interference dominates, shrinks the offset to at most $1/(2P)$ bins and reduces the loss in Table~\ref{tab:round} to 0.014~dB or less for $P\ge2$, but it is not the same processor. The FFT grows to $P^2N$
points, about 4.8 times the cost of the base FFT for $P=2$ against 20\% for the ramp, and a
$W\times W$ window on the padded grid spans only $W/P$ native beams, so it sees less of the
interference around the look direction. Keeping the angular extent needs $P^2B$ beams, and the adaptive stage, whose weight update takes on the order of $B^3/3+KB^2$ operations, grows by a factor approaching $P^6$ for large $B$. The loss in Table~\ref{tab:round} is a property of the look
vector alone; with interference present, the subspace loss also depends on where the jammers sit
relative to the window, which is what the simulation below measures.

\section{Placement Rules Against Jamming}

\subsection{Setup}
The scenario is jamming only: three noise jammers at 30~dB per element, at $(u,v)=(0.675,0)$,
$(0.283,0.05)$ and $(0.600,0.35)$, isotropic elements, 0~dB signal-to-noise ratio per element,
and no clutter. The look direction is $(u_s,0)$ with $u_s=0.20$, 0.55 and 0.85. The carrier takes 41 evenly spaced values from 4 to 8~GHz, 500 dwells each, and every dwell is scored for the three placement rules and for windows of 9, 25 and 49 beams. The training support is
$K=2B$ snapshots (18, 50 and 98) with diagonal loading at the noise power. The loss is the
output SINR divided by the element-space optimum. The \emph{worst carrier} is the carrier at which
the mean loaded-SMI loss over the 500 dwells is lowest. The interval stored with each mean has a half-width of at most 0.04~dB for the fractional and nearest-bin rules and at most 0.05~dB for the frozen window. The testbench was built in Simulink from library blocks, and every environment model in it was validated against a separate implementation.

\subsection{Results}
Table~\ref{tab:e1} gives the worst-carrier loss for the three rules. The frozen window is not
viable away from broadside. At $u_s=0.55$ it loses 39 to 60~dB at the worst carrier, and the
clairvoyant subspace loss exceeds 3~dB at 95, 66 and 41\% of the carriers for 9, 25 and 49
beams; at $u_s=0.85$ the shares are 76, 61 and 46\%. This is the geometry of
Figure~\ref{fig:drift}: once the look direction leaves the window, the processor sees it only
through sidelobes. Near broadside the drift is small (1.6 bins) and a frozen $5\times5$ or
$7\times7$ window is within 0.15~dB of the others.

\begin{table}[!t]
\caption{Worst-carrier loaded-SMI loss (\textnormal{dB}), $K=2B$, three 30~\textnormal{dB} jammers. Gain is fractional minus nearest-bin, so a positive gain favours the exact window. The last column is the share of the 41 carriers at which the clairvoyant subspace loss of the frozen window exceeds 3~\textnormal{dB}.}
\label{tab:e1}
\centering
\footnotesize
\setlength{\tabcolsep}{3pt}
\begin{tabular}{@{}ccrrrrr@{}}
\toprule
$u_s$ & $B$ & Fractional & Nearest & Frozen & Gain & Frozen $>3$~dB \\
\midrule
0.20 & 9  & $-4.13$ & $-3.52$ & $-5.13$  & $-0.61$ & 0.07 \\
0.20 & 25 & $-1.87$ & $-1.80$ & $-1.82$  & $-0.07$ & 0.00 \\
0.20 & 49 & $-1.07$ & $-1.19$ & $-1.22$  & $+0.12$ & 0.00 \\
\midrule
0.55 & 9  & $-2.49$ & $-3.14$ & $-59.55$ & $+0.65$ & 0.95 \\
0.55 & 25 & $-1.01$ & $-1.37$ & $-49.47$ & $+0.36$ & 0.66 \\
0.55 & 49 & $-0.63$ & $-0.94$ & $-38.71$ & $+0.31$ & 0.41 \\
\midrule
0.85 & 9  & $-2.20$ & $-2.47$ & $-52.12$ & $+0.27$ & 0.76 \\
0.85 & 25 & $-1.21$ & $-1.31$ & $-42.43$ & $+0.10$ & 0.61 \\
0.85 & 49 & $-0.81$ & $-0.95$ & $-32.75$ & $+0.14$ & 0.46 \\
\bottomrule
\end{tabular}
\end{table}

The exact window beats nearest-bin placement in seven of the nine cases, by 0.10 to 0.65~dB,
most for the smallest window and the look direction farthest from broadside: 0.65, 0.36 and
0.31~dB at $u_s=0.55$ for 9, 25 and 49 beams. It does not win everywhere. At $u_s=0.20$ the
nine-beam fractional window is 0.61~dB worse at its worst carrier than the nearest-bin window,
and the 25-beam window is 0.07~dB worse, a small difference that is nevertheless larger than the combined half-width of the stored intervals (about 0.03~dB). At that setting
the fractional window has a subspace loss beyond 3~dB at two of the 41 carriers (share 0.05),
while the nearest-bin window has none. Exact capture of the look vector is therefore the
property that the fractional rule guarantees; it is not by itself the property that minimises
the loss when the window is small and a jammer lies near it, because a window shifted by up to half a bin may also hold more of that jammer's leakage.

Figure~\ref{fig:loss} shows the clairvoyant subspace loss against carrier for the 25-beam window at $u_s=0.55$ and the nine-beam window at $u_s=0.20$.

\begin{figure}[!t]
\centering
\includegraphics[width=\columnwidth]{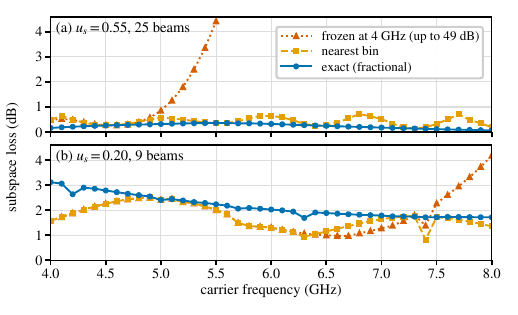}
\caption{Clairvoyant subspace loss against carrier for the three placement rules and the three jammers of Table~\ref{tab:e1}. (a) 25 beams at $u_s=0.55$: the exact window peaks at 0.36~\textnormal{dB} and the nearest-bin window at 0.72~\textnormal{dB}, while the frozen window exceeds 3~\textnormal{dB} from 5.4~\textnormal{GHz} and reaches 49~\textnormal{dB} at 8~\textnormal{GHz}. (b) Nine beams at $u_s=0.20$: the exact window has a larger loss than the nearest-bin window at 38 of the 41 carriers.}
\label{fig:loss}
\end{figure}

\subsection{Random laydowns and jammer-aware placement}
The laydown of Table~\ref{tab:e1} is one draw, so the ranking of the rules is also checked on random ones. The array, the 41 carriers and the 30~dB jammer level are those of Section~IV-A. For each look direction, 200 laydowns of three jammers are drawn with directions uniform in a disc of radius 0.5 around the look direction in $(u,v)$, at least 0.06 from it (one bin at 8~GHz, half a bin at 4~GHz), so that no jammer coincides with the look direction, where the optimum collapses, and inside $u^2+v^2\le0.95^2$. Every laydown is scored at all 41 carriers, which gives 8200 cases per row, with the true covariance, so that sampling noise does not enter. The script behind Table~\ref{tab:lay} reproduces the 27 stored worst-carrier clairvoyant losses of the simulation behind Table~\ref{tab:e1} to within 0.001~dB, which supports using it away from the stored laydown.

The exact window has the smaller loss than the nearest-bin window in 70 to 95\% of the cases (Table~\ref{tab:lay}), and the median loss of the exact 25- and 49-beam windows is 0.15~dB or less. The lead is smallest for nine beams (70 to 75\%), which is consistent with the reversal for nine beams near broadside in Table~\ref{tab:e1}. The frozen window has a loss above 3~dB in 32 to 80\% of the cases at $u_s=0.55$ and 0.85, and in at most 2\% at $u_s=0.20$.

Section~V raises the possibility of choosing the window centre from the jammer directions. To measure what it could gain, each case in a subset (60 laydowns, 11 carriers, 660 cases per row) is also scored at 81 centres within one bin of the look bin, on a grid of 0.25 bin that includes the exact centre. This is an oracle benchmark: it uses the true covariance, and no implementable rule is given here. The best of them improves on the exact window by a median of 0.45 to 0.65~dB for nine beams, 0.06~dB for 25 beams and 0.01~dB for 49 beams. The 90th percentile is 1.3, 0.44 and 0.19~dB, and the largest gain is 4.2~dB, for nine beams. For 25 beams or more, centring on the look direction is therefore within 0.44~dB of the best centre in nine cases out of ten.

\begin{table}[!t]
\caption{Clairvoyant subspace loss over random three-jammer laydowns (8200 cases per row). Median loss (\textnormal{dB}, positive); share of cases (\%) in which the exact window beats the nearest-bin window and in which the frozen window loses more than 3~\textnormal{dB}; gain (\textnormal{dB}) of the best centre within one bin of the look bin over the exact window (660 cases per row).}
\label{tab:lay}
\centering
\footnotesize
\setlength{\tabcolsep}{2.2pt}
\begin{tabular}{@{}ccccccccc@{}}
\toprule
 & & \multicolumn{3}{c}{Median loss} & \multicolumn{2}{c}{Share (\%)} & \multicolumn{2}{c}{Gain} \\
\cmidrule(lr){3-5}\cmidrule(lr){6-7}\cmidrule(l){8-9}
$u_s$ & $B$ & Exact & Nearest & Frozen & \shortstack{Exact\\better} & \shortstack{Frozen\\$>3$} & Median & \shortstack{90th\\pct} \\
\midrule
0.20 & 9 & 0.86 & 0.95 & 1.05 & 70 & 2 & 0.52 & 1.15 \\
 & 25 & 0.11 & 0.30 & 0.31 & 85 & 0 & 0.06 & 0.43 \\
 & 49 & 0.02 & 0.15 & 0.15 & 91 & 0 & 0.01 & 0.19 \\
\midrule
0.55 & 9 & 0.81 & 0.98 & 18.20 & 75 & 80 & 0.45 & 1.16 \\
 & 25 & 0.11 & 0.35 & 5.89 & 90 & 55 & 0.06 & 0.40 \\
 & 49 & 0.01 & 0.20 & 0.41 & 95 & 32 & 0.01 & 0.17 \\
\midrule
0.85 & 9 & 1.05 & 1.14 & 19.98 & 71 & 77 & 0.65 & 1.31 \\
 & 25 & 0.15 & 0.37 & 14.30 & 87 & 61 & 0.06 & 0.44 \\
 & 49 & 0.02 & 0.19 & 0.68 & 93 & 47 & 0.01 & 0.18 \\
\bottomrule
\end{tabular}
\end{table}

\section{Discussion}

Three conclusions follow, with their limits. First, a hopping radar cannot freeze its beamspace
window unless its look direction stays close to broadside and the window is large compared with
the drift. Second, nearest-bin placement is a small penalty, 0.1 to 0.65~dB in these
cases and up to 1.34~dB of look capture in the worst case, and the exact rule removes it for the
price of a phase ramp that costs about a fifth of the FFT; zero padding all but removes the look-capture loss too, but changes the processor and costs more. Third, exact placement is not always the
loss-minimising placement: the placement that minimises the subspace loss depends on where the
jammers are, which suggests choosing the window centre from the interference geometry as well
as from the look direction. Section~IV-C measures what that could gain: for 25 beams or more it is 0.44~dB or less in nine cases out of ten, so the look-centred exact rule is close to the best centre, and only the nine-beam window gains much, up to 4.2~dB in the worst of 660 cases. That is an oracle figure. A rule that chose the centre from the training data would make the window data-dependent, which none of the results here covers.

The sampled-covariance evidence (Table~\ref{tab:e1}) comes from one jamming laydown, and the random-laydown evidence (Table~\ref{tab:lay}) uses the true covariance, three jammers of equal power within 0.5 of the look direction, and one drawing scheme. Both are for three look directions with $v_s=0$, one array size, one loading level, an ideal array, and a square window. Land
clutter, detection, array errors, multipath and the comparison with a zero-padded grid in
simulation, as opposed to the look-capture analysis of Section~III, are outside it. The stored intervals of the two compared rules are at most 0.04~dB per estimate, so the smallest differences reported (0.07 and 0.10~dB) exceed them, but they are too small to matter in practice.

\section{Conclusion}

Because the beamspace transform is orthonormal, a hop affects a windowed beamspace processor
only through the position of the window. Placing the window's centre at the fractional bin where the look direction falls keeps the look vector fully captured at every carrier, for the cost of one phase ramp. A window frozen at one carrier fails for any look direction far from broadside, with losses of 33 to 60~dB at the worst carrier for $u_s=0.55$ and 0.85, and rounding to the nearest bin
costs up to 1.34~dB of look capture and 0.1 to 0.65~dB of loss against three jammers, except for small windows near broadside, where it did slightly better. Over random jammer laydowns the exact window had the smaller clairvoyant loss in 70 to 95\% of the cases, and moving the window by up to a bin with oracle knowledge of the jammers gained a median of at most 0.65~dB for nine beams and 0.06~dB for 25 beams or more.

\section*{Acknowledgment}
The authors used Claude (Anthropic), a large language model, to help prepare this version of the paper from their earlier draft. The authors provided the results and the derivations, Claude helped to format the paper for the conference template and to proof-read it. The authors reviewed and accepted all changes, approved the final text, and take full responsibility for the content.

\end{document}